\documentclass[11pt,a4paper,reqno]{amsart}

\usepackage[T1]{fontenc}
\usepackage[utf8]{inputenc}
\usepackage{lmodern}
\usepackage{amsmath,amssymb,mathtools}
\usepackage{microtype}
\usepackage[colorlinks=true,linkcolor=blue,citecolor=blue,urlcolor=blue]{hyperref}
\hypersetup{
 pdftitle={Szeg\H{o}-type determinant asymptotics for generalized Hilbert matrices with edge eigenvalues},
 pdfauthor={Martin Gebert and Heinrich K\"uttler}
}
\usepackage{pgfplots}

\newtheorem{theorem}{Theorem}[section]
\newtheorem{proposition}[theorem]{Proposition}
\newtheorem{lemma}[theorem]{Lemma}
\newtheorem{corollary}[theorem]{Corollary}
\theoremstyle{remark}
\newtheorem{remark}[theorem]{Remark}

\newcommand{\Nzero}{\mathbb N_0}
\newcommand{\C}{\mathbb C}
\newcommand{\R}{\mathbb R}

\newcommand{\Hh}{\mathcal H}
\newcommand{\Tr}{\operatorname{Tr}}
\newcommand{\Ker}{\operatorname{Ker}}
\newcommand{\Ran}{\operatorname{Ran}}
\newcommand{\dist}{\operatorname{dist}}
\newcommand{\id}{I}
\newcommand{\cS}{\mathfrak S_1}
\DeclarePairedDelimiter{\abs}{\lvert}{\rvert}
\DeclarePairedDelimiter{\norm}{\lVert}{\rVert}
\newcommand{\minushalf}{{\nulldelimiterspace=0pt $\vcenter{\hbox{$\scriptstyle-\frac12$}}$}}
\newcommand{\plushalf}{{\nulldelimiterspace=0pt $\vcenter{\hbox to .5em{\hss$\scriptstyle\frac12$\hss}}$}} % digit-wide box: centered under 0 and 1 on the y axis
\newcommand{\minushalfy}{{\nulldelimiterspace=0pt $\vcenter{\hbox{$\scriptstyle-\hbox to .5em{\hss$\scriptstyle\frac12$\hss}$}}$}} % same, with minus sign, aligned like -1 on the y axis

\pgfplotsset{set layers}
\pgfplotsset{compat=1.18}

\let\<\langle
\let\>\rangle

\title[Szeg\H{o}-type asymptotics for generalized Hilbert matrices]{Szeg\H{o}-type determinant asymptotics for generalized Hilbert matrices with edge eigenvalues}
\author{Martin Gebert}
\address{(M.~Gebert) Munich, Germany}
\email{gebert@math.lmu.de}
\author{Heinrich K\"uttler}
\address{(H.~K\"uttler) Palo Alto, CA, USA}
\email{kuettler@math.lmu.de}
\date{}

\begin{document}

\begin{abstract}
We compute the large-$N$ asymptotics of the determinant of the
identity plus or minus a generalized Hilbert matrix.
For $N \in \mathbb{N}$ and $\delta \in \mathbb{R} \setminus \{-1, -2,
  \ldots\}$, let
\[
  H_N^\delta
  := \left(
    \frac{\sin((1+\delta)\pi)}{\pi(j+k+1+\delta)}
  \right)_{j,k=0}^{N-1}
\]
be the generalized Hilbert matrix. For $\delta$ not a half-integer, we prove the large-$N$ power-law
asymptotics of the determinant
\[
  \log \det(\id_N \pm H_N^\delta) = -\frac{\theta_\delta^2 \pm \theta_\delta}{2} \log N + O(1)
\]
where
\[
  \theta_\delta :=
  \begin{cases}
    \delta, & \delta < -\frac{1}{2}, \\
    \dfrac{1}{\pi} \arcsin\bigl(\sin(\delta\pi)\bigr), & \delta \geq -\frac{1}{2},
  \end{cases}
\]
and $\arcsin \colon [-1,1] \to [-\frac{\pi}{2}, \frac{\pi}{2}]$ denotes the principal branch of $\arcsin$. For $\delta \geq -\frac 12$ the asymptotics is known. The novelty of the present paper is the regime $\delta < - \frac 12$ and understanding the dichotomy in the decay exponent. This stems from the $\pm 1$ edge eigenvalues of the limiting operator appearing for $\delta < - \frac 12$.
Most notably, we obtain for $\delta < -\frac 12$
\[
  \log \det\bigl(\id_N - (H_N^\delta)^2\bigr) = -\delta^2\, \log
  N + O(1).
\]
Such determinants arise in the study of Anderson's orthogonality catastrophe.
\end{abstract}

\maketitle

\section{Introduction and results}

For $\alpha\notin\{0,-1,-2,\ldots\}$, we define
the bounded self-adjoint \emph{generalized Hilbert matrix} on
$\ell^2(\Nzero)$ by
\begin{equation}\label{eq:Halpha}
 \Hh_\alpha:=\left(\frac1{j+k+\alpha}\right)_{j,k\in\Nzero}
\end{equation}
 and its normalized form by
\begin{equation}\label{eq:Aalpha}
 A_\alpha:=\frac{\sin(\pi\alpha)}\pi\Hh_\alpha.
\end{equation}
For $\delta\in\mathbb R \setminus \{-1,-2,\ldots\}$, let
\begin{equation}
 H^\delta := A_{1+\delta}
\end{equation}
and for $N\in\mathbb N$, we denote its $N\times N$ restriction by
\begin{equation}\label{eq:def-Hdelta}
 H_N^\delta
 :=
 \Big(H^\delta(j,k)\Big)_{j,k=0}^{N-1}.
\end{equation}

Our goal is to understand the asymptotic behavior of the determinants
\begin{equation}\label{eq:def-Dpm}
 D_N^\pm(\delta):=\det(\id_N\pm H_N^\delta)
\end{equation}
as the matrix size $N\to\infty$, where $\id_N$ denotes the $N\times N$ identity. We note that $D_N^\pm(\delta)>0$ by Lemma~\ref{lem:positivity} (for $\delta\in\Nzero$ trivially, since then $H_N^\delta=0$) and therefore $\log D_N^\pm(\delta)$ is well-defined.

\begin{figure}
\centering
\begin{tikzpicture}
\begin{axis}[
    width=0.97\linewidth, height=5.2cm,
    xlabel={$\delta$},
    ylabel={$\theta_\delta$},
    xmin=-4, xmax=6, ymin=-4, ymax=1,
    xtick={-4,-3,...,6},
    minor x tick num=1,
    extra x ticks={-0.5},
    extra x tick labels={\minushalf},
    extra x tick style={grid=none},
    ytick={-4,-3,-2,-1,0,1},
    extra y ticks={-0.5,0.5},
    extra y tick labels={\minushalfy,\plushalf},
    axis lines=left,
    axis line style={gray!60},
    tick style={gray!60},
    tick label style={font=\footnotesize},
    x tick label style={text height=1.5ex, text depth=.25ex},  % uniform boxes -> aligned baselines
    grid=both,
    major grid style={gray!25},
    minor grid style={gray!12},
    ylabel style={rotate=-90},
    clip=false,
]

% Shade the region delta < -1/2 (the non-folded branch)
\addplot[draw=none, fill=black!5, forget plot, on layer=axis background]
    coordinates {(-4,-4) (-0.5,-4) (-0.5,1) (-4,1)} \closedcycle;

% The function is piecewise linear, so exact vertices beat sampling.
% delta <= 1/2 : the line delta    (the two branches agree on [-1/2, 1/2])
% delta >  1/2 : triangle wave, reflecting at every half-integer
% Domain is [-4, 6]; the endpoint (6, 0) falls mid-tooth.
\addplot[very thick, blue!55!black] coordinates {
    (-4, -4) (0.5, 0.5)
    (1.5, -0.5) (2.5, 0.5) (3.5, -0.5) (4.5, 0.5)
    (5.5, -0.5) (6, 0)
};

\end{axis}
\end{tikzpicture}
\caption{$\theta_\delta$ as a function of $\delta$}
\label{fig:effective-phase}
\end{figure}
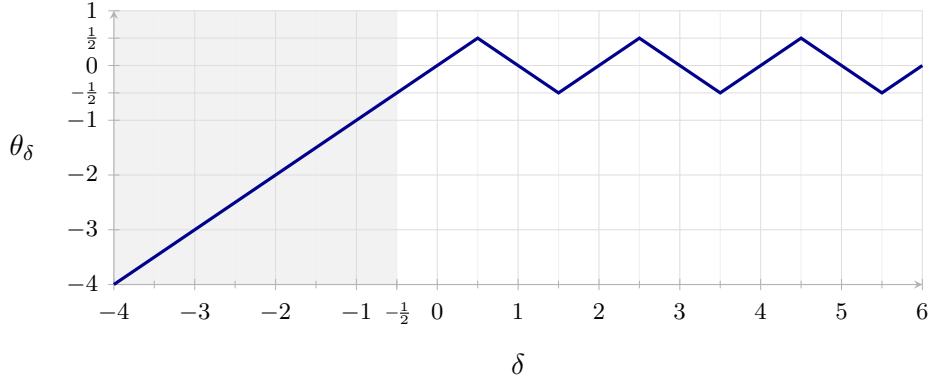

Define the pole and critical sets
\begin{equation}\label{eq:exceptional-sets}
 \mathcal P:=\{-1,-2,\ldots\}
 \qquad \text{and} \qquad
 \mathcal C:=\frac12+\mathbb Z.
\end{equation}

Our main result is the following:

\begin{theorem}\label{thm:main}
Let $\delta\in\mathbb R\setminus(\mathcal P\cup\mathcal C)$. We define an $\arcsin$-\,$\sin$ regularization of $\delta$ by
\begin{align}\label{eq:def-sq}
 \delta_{\text{reg}}:=&\frac1\pi\arcsin\bigl(\sin(\delta\pi)\bigr)
 \in\left[-\frac12,\frac12\right]
\end{align}
where $\arcsin:[-1,1]\to[-\pi/2,\pi/2]$ is the principal inverse sine. With the latter, we further define
\begin{equation}\label{eq:effective-phase}
 \theta_\delta:=
 \begin{cases}
  \delta,&\delta < -\frac12,\\[1mm]
  \delta_{\text{reg}},&\delta \geq -\frac12.
 \end{cases}
\end{equation}
Then, we have the following asymptotic behavior as $N\to\infty$
\begin{equation}\label{eq:strong-main}
 \log D_N^\pm(\delta)
 =-\frac{\theta_\delta^2\pm\theta_\delta}{2}\log N+O_\delta(1).
\end{equation}
\end{theorem}

\begin{remark}
(1) For $\delta \geq -\frac12$, the asymptotics is already
proven in \cite{Otte}; with the weaker remainder $o(\log N)$ it was also
obtained earlier, in the more general framework of Hankel matrices whose symbol
has jump discontinuities \cite{FedeleGebert}. At a high level, the asymptotics arise from the band of absolutely continuous spectrum between \(0\) and \(\sin((1+\delta)\pi)\), see Section~\ref{sec:spectrum}. We call this the \emph{bulk contribution} to the asymptotics.

(2) The novelty of the present
paper is the case $\delta < -\frac12$, where the operator $H^\delta$, in addition to the absolutely continuous spectrum, acquires
eigenvalues at $\pm1$; for its spectral structure, see
Section~\ref{sec:spectrum}. The main content of the paper is to control the
effect of these edge eigenvalues on $D_N^\pm(\delta)$. We call this the \emph{edge contribution}. Notably, the bulk and edge contributions
 combined reproduce the same closed expression
\eqref{eq:strong-main}, now with the unfolded exponent $\theta_\delta=\delta$
(rather than $\delta_{\mathrm{reg}}$).

(3) Note that the exponent $\delta\mapsto\theta_\delta$ is a continuous function. See
Figure~\ref{fig:effective-phase} for an illustration.
Moreover,
\begin{equation}\label{eq:delta-reg-s}
  \delta_{\text{reg}}
    = \frac1\pi\arcsin\bigl(\sin(\delta\pi)\bigr)
    = - \frac1\pi\arcsin\bigl(\sin((1+\delta)\pi)\bigr)
\end{equation}
and $\abs{\delta_{\text{reg}}}=\dist(\delta,\mathbb Z)$.
\end{remark}

\begin{corollary}\label{cor:main-sq}
Let $\delta\in\mathbb R\setminus(\mathcal P\cup\mathcal C)$. Then
\begin{equation}\label{eq:strong-main-sq}
 \log\det\bigl(\id_N-(H_N^\delta)^2\bigr)
 =-\theta_\delta^2\log N+O_\delta(1)
\end{equation}
as $N\to\infty$.
\end{corollary}

\begin{remark}[Critical half-integers]
At the critical half-integers $\delta\in\mathcal C$ excluded from
Theorem~\ref{thm:main}, the strong remainder $O_\delta(1)$ degrades but the
leading term survives:
\begin{equation}\label{eq:critical-weak}
 \log D_N^\pm(\delta)
 =-\frac{\theta_\delta^2\pm\theta_\delta}{2}\log N+o(\log N).
\end{equation}
For $\delta\ge-\frac12$, this is proven in \cite{Otte}.
For $\delta<-\frac12$, the proof requires a separate threshold
argument, which we only sketch here to keep the presentation
concise: deleting finitely many initial rows and columns
reduces the question to the matrix $\id_N\pm H_N^{-1\mp1/2}$, and a
Schur-complement argument shows that the decay of its smallest eigenvalue
changes the logarithm of the determinant only by \(o(\log N)\), so that the
noncritical argument still applies.
\end{remark}

Our main motivation is to understand the contribution of the edge eigenvalues to
the determinant asymptotics. Determinants of this type have, however, already
appeared in a number of other mathematical contexts:

\begin{remark}\label{rem:motivation}
(1) \emph{Anderson's orthogonality catastrophe.} Determinants of the form
$D_N^\pm(\delta)$ appear in a phenomenon called
Anderson's orthogonality catastrophe (AOC), which describes the
asymptotic vanishing of ground-state overlaps of two non-interacting Fermi
gases, see \cite{GebertKuttlerMuller,GebertKuttlerMullerOtte}. The essential
task in AOC is to understand the precise $L\to\infty$ asymptotics of determinants
involving differences of spectral projections of
the form
\begin{equation} \label{eq:aoc-det}
    \det\bigl( \id - \mathbf 1_L ( \mathbf 1_{(-\infty, E]} (H_0) - \mathbf 1_{(-\infty, E]} (H_1) )^2 \mathbf 1_L\bigr),
\end{equation}
for appropriate pairs of Schr\"odinger operators $(H_0, H_1)$ which differ by a short-range scattering potential and
finite-volume projections $\mathbf 1_L$.

Products and differences of such spectral projections
are closely related to Hankel operators \cite{FrankPushnitski}, and the generalized Hilbert matrix
treated here can serve as a model operator, since it often plays a key
role in the analysis. For example, the generalized Hilbert matrix already appeared in the
study of Krein's example which is a simple pair of such spectral
projections, see \cite{KostrykinMakarov}. Powers of the Hilbert matrix
also yield the coefficients of a series expansion of the AOC decay
exponent in \cite{GebertKuttlerMullerOtte}.

In the direct context of AOC, the determinant $\det(\id_N-(H_N^\delta)^2)$ appeared
in \cite{KnoerrOtteSpitzer}. Moreover, AOC for Dirac-$\delta$ perturbations is
treated in \cite{Gebert2015}, where the same phase transition of the decay
exponent was observed as in \eqref{eq:effective-phase}. In this case, $-\pi\delta$ is the scattering phase shift.

To prove the $L$-asymptotics of \eqref{eq:aoc-det}, one could follow the strategy of this paper and decompose the determinant into a bulk and an edge determinant. The bulk asymptotics is known \cite{FrankPushnitski}; the edge contribution would require the tail asymptotics of the eigenfunctions corresponding to the edge eigenvalues $\pm1$ to determine the exact AOC asymptotics.

(2) \emph{Appearance in random matrix theory.} The determinant
$\det(\id_N - H_N^{-1/2})$ equals the probability that a corank-one truncation of
a $(2N{+}1)$-dimensional random orthogonal matrix has no real eigenvalues. This
is closely related to the persistence probability that a Kac random polynomial
has no real zeros \cite{GebertPoplavskyi}.

(3) \emph{Asymptotics of Hankel, Toeplitz, and Hankel-plus-Toeplitz determinants}
is a classical subject. Since $\id_N$ is (trivially) Toeplitz and $H_N^\delta$
is Hankel, $D_N^\pm(\delta)$ is itself a Hankel-plus-Toeplitz determinant.
Classical formulas typically require a specific relation between the Toeplitz and
Hankel symbols \cite{BasorEhrhardt,DeiftItsKrasovsky}. More recent
Riemann--Hilbert approaches allow the two symbols to be distinct, but only under
regularity assumptions that are not satisfied in our setting
\cite{GharakhlooIts,GharakhlooItsII}.
\end{remark}

\section{The spectrum of the generalized Hilbert matrix}\label{sec:spectrum}

Let $\ell^2(\Nzero)$ have standard orthonormal basis $(e_j)_{j\ge0}$.  Throughout let $P_N$ be the orthogonal projection onto
$\operatorname{span}\{e_0,\ldots,e_{N-1}\}$ and put
\begin{equation}
    Q_N:=\id-P_N.
\end{equation}
We remark that
\begin{equation}\label{eq:compression}
 H_N^\delta=P_NA_{1+\delta} P_N\big|_{P_N\ell^2(\Nzero)}.
\end{equation}
If $\delta\in\Nzero$, then $\sin((1+\delta)\pi)=0$, so $H^\delta=0$ and
$D_N^\pm(\delta)=1$; we assume $\delta\notin\mathbb Z$ from now on.

The generalized Hilbert matrix $H^\delta = A_{1+\delta}$ can be diagonalized explicitly, see e.g., \cite{AlemanMontesSarafoleanu,RosenblumII,Silbermann}:

\begin{proposition}[Spectrum of generalized Hilbert matrix $H^\delta$]\label{prop:spectral-data}
Let $\delta\in\R\setminus\mathbb Z$ and set
$s:=\sin((1+\delta)\pi)$.  Then:

\begin{enumerate}
\item The absolutely continuous spectrum of $H^\delta$ is
\begin{equation}\label{eq:ac-spectrum}
 \operatorname{spec}_{\mathrm{ac}}(H^\delta)
 =[\min\{0,s\},\max\{0,s\}],
\end{equation}
with multiplicity one, i.e., there is one band of absolutely continuous spectrum.

\item The point spectrum of $H^\delta$ is contained in $\{-1,+1\}$.  Define
\begin{align}
 J^-_\delta
 &:=\left\{j\in 2\Nzero:0\leq j<-\delta-\frac12\right\},
 \label{eq:Jminus}\\
 J^+_\delta
 &:=\left\{j\in 2 \Nzero +1:0\leq j<-\delta-\frac12\right\}.
 \label{eq:Jplus}
\end{align}
For each $j\in J^-_\delta \cup J^+_\delta$, let $u_j$ be the Taylor coefficient vector of
\begin{equation}\label{eq:fj}
 f_j(z):=(1-z)^{-(1+\delta+j)}(1+z)^j.
\end{equation}
Then
\begin{equation}\label{eq:eigenvalue-fj}
 H^\delta u_j=(-1)^j u_j,
\end{equation}
and
\begin{equation}\label{eq:kernel-basis}
 \Ker(\id\pm H^\delta)
 =\operatorname{span}\{u_j:j\in J^\pm_\delta\},
\end{equation}
i.e., $J_\delta^-$ indexes the eigenvectors for the eigenvalue $+1$ and
$J_\delta^+$ those for the eigenvalue $-1$.

\item
The singular continuous spectrum is empty, i.e., the above describes the spectrum of $H^\delta$ completely.
\end{enumerate}
\end{proposition}

\begin{remark}
(1) \emph{Symbol of the generalized Hilbert matrix.} For a Hankel matrix $H:\ell^2(\Nzero)\to\ell^2(\Nzero)$, we call a bounded function
$f:\mathbb T\to\mathbb C$ a \emph{symbol of $H$} if
\begin{equation}
 H(j,k)=\hat f(j+k), \qquad j,k\in\Nzero,
\end{equation}
where $\mathbb T:=\{z\in\mathbb C:\abs{z}=1\}$ and
\begin{equation}
 \hat f(n):=\frac1{2\pi}\int_0^{2\pi}f(e^{it})e^{-int}\,dt, \qquad n\in\Nzero,
\end{equation}
denotes the $n$-th Fourier coefficient of $f$. For the operator $H^\delta$
studied here, with $\alpha:=1+\delta$, the function
\begin{equation}
 f(e^{it})=e^{i\alpha(\pi-t)}, \qquad t\in[0,2\pi),
\end{equation}
is a symbol. It is unimodular and continuous except for a single jump at
$z=1$, where
\begin{equation}
 f(1^\pm)=e^{\pm i\alpha\pi},\qquad
 \tfrac{1}{2i}\bigl(f(1^+)-f(1^-)\bigr)=\sin(\pi\alpha).
\end{equation}
By the theory of piecewise-continuous Hankel symbols, this jump
 accounts for the absolutely continuous spectrum \eqref{eq:ac-spectrum}, see \cite{Howland}.

(2) The operator $H^\delta$ is a self-adjoint contraction. If $\abs{s}<1$,
then $\operatorname{spec}(\id\pm H^\delta)\setminus\{0\}\subseteq[1-\abs{s},2]$;
in particular, $\id\pm H^\delta\ge1-\abs{s}$ on the orthogonal complement of
the kernel \eqref{eq:kernel-basis}.

(3) The asymptotic behavior of $u_j(n)$ as $n\to\infty$, derived in
Lemma~\ref{lem:eigenvector-tails} below, shows that
$u_j\in\ell^2(\Nzero)$. The $u_j$ form a basis of $\Ker(\id\pm H^\delta)$ but are neither normalized nor orthogonal.
\end{remark}

\begin{lemma}\label{lem:positivity}
Let $N\in\mathbb N$ and $\delta\in\mathbb R\setminus\mathbb Z$. Then
\begin{equation}
 \id_N\pm H_N^\delta>0,
\end{equation}
and consequently $D_N^\pm(\delta)=\det(\id_N\pm H_N^\delta)>0$, so that
$\log D_N^\pm(\delta)$ is well-defined.
\end{lemma}

\begin{proof}
By \eqref{eq:compression}, $\langle\psi,(\id_N\pm H_N^\delta)\psi\rangle
=\langle\psi,(\id\pm H^\delta)\psi\rangle$ for any $\psi\in P_N\ell^2(\Nzero)$.

Suppose there exists $0\ne\psi\in P_N\ell^2(\Nzero)$ such that
\begin{equation}
    \langle\psi,(\id\pm H^\delta)\psi\rangle = 0.
\end{equation}
With
$B:=(\id\pm H^\delta)^{1/2}\ge0$, the above shows $\norm{B\psi}^2=0$ and $B\psi=0$, hence
$(\id\pm H^\delta)\psi=B(B\psi)=0$, i.e., $\psi\in\Ker(\id\pm H^\delta)$ and $H^\delta$ has an eigenvalue $\mp 1$ with compactly supported eigenfunction.

If $\delta\ge-\frac12$, Proposition~\ref{prop:spectral-data} shows that
$H^\delta$ has no eigenvalues at all, so $\Ker(\id\pm H^\delta)=\{0\}$ by
\eqref{eq:kernel-basis}, contradicting $\psi\ne0$.

If $\delta<-\frac12$, we show that there cannot be a $\mp 1$ eigenvalue with compact support. Note that \eqref{eq:kernel-basis} identifies
$\Ker(\id\pm H^\delta)$ with the span of $\{u_j:j\in J_\delta^\pm\}$. Writing
$\psi=\sum_{j\in J_\delta^\pm}\gamma_ju_j$ and
$j^*:=\max\{j\in J_\delta^\pm:\gamma_j\ne0\}$, \eqref{eq:coeff-asymptotic}
gives $\psi(n)=\gamma_{j^*}c_{j^*}n^{\delta+j^*}(1+o(1))$ with
$\gamma_{j^*}c_{j^*}\ne0$, so $\psi(n)\ne0$ for large $n$, contradicting
$\psi\in P_N\ell^2(\Nzero)$.
\end{proof}

\section{The bulk and the edge determinant}\label{sec:bulk-edge}

In this section we decompose the determinant $D_N^\pm(\delta)$ into two determinants:
\begin{itemize}
    \item[(i)] (Bulk determinant) A determinant involving a matrix which is based only on the absolutely continuous part of  $H^\delta$.
    \item[(ii)] (Edge determinant) A determinant involving (mainly) the eigenspaces of the edge eigenvalues $\pm1$ of $H^\delta$.
\end{itemize}
To do so, we now fix a noncritical parameter $\delta\in\mathbb R\setminus(\mathcal P\cup\mathcal C)$, set
\begin{equation}\label{eq:nonthreshold}
 s:=\sin((\delta +1)\pi),\qquad \abs{s}<1,
\end{equation}
and define
\begin{equation}\label{eq:Tpm}
 T^\pm:=\id\pm H^\delta,
 \qquad
 K^\pm:=\Ker T^\pm.
\end{equation}
Let $\Pi^\pm$ be the orthogonal projection onto $K^\pm$, and set
\begin{equation}\label{eq:Ttilde-pm}
 \widetilde T^{\,\pm}:=T^\pm+\Pi^\pm.
\end{equation}
By Proposition~\ref{prop:spectral-data}, the spectrum of
$\widetilde T^{\,\pm}$ lies in $[1-\abs{s},2]$, and thus
$\widetilde T^{\,\pm}\ge1-\abs{s}$.

We decompose $D_N^\pm(\delta)$ using the following lemma:

\begin{lemma}[Bulk and edge determinant factorization]\label{lem:factorization}
Let $T\ge0$ be bounded on $\ell^2(\Nzero)$.  Suppose that
$K:=\Ker T$ has finite dimension $m$ and that $T$ is bounded below by a
positive constant on $K^\perp$.  Let $\Pi$ be the orthogonal projection onto
$K$, set $\widetilde T:=T+\Pi$, and let $V$ be an isometry $\C^m\to K$.  Then
\begin{equation}\label{eq:edge-factorization}
 \det(P_NTP_N)
 =\det(P_N\widetilde TP_N)\det M_N,
\end{equation}
where
\begin{equation}\label{eq:def-MN}
 M_N:=\id_m-V^*P_N(P_N\widetilde TP_N)^{-1}P_NV
\end{equation}
and $\id_m$ denotes the identity on $\mathbb C^m$.
\end{lemma}

\begin{proof}
The assumptions imply that $\widetilde T$ has a positive lower bound
and is therefore invertible. On $P_N(\ell^2(\Nzero))$, the
finite matrix $P_N\widetilde T P_N$ satisfies the same lower bound and is therefore invertible as well.
Since $\Pi=VV^*$, we obtain
\begin{align}
    P_NTP_N &= P_N\widetilde TP_N-(P_NV)(P_NV)^*\notag \\
    &= P_N\widetilde TP_N( \id - (P_N\widetilde TP_N)^{-1}(P_NV)(P_NV)^*).
\end{align}
Since $\det(\id - ABC) = \det(\id - CAB)$, we obtain \eqref{eq:edge-factorization} and \eqref{eq:def-MN}.
\end{proof}

Applying the above lemma to $T^\pm$, we obtain
\begin{equation}
    D_N^\pm(\delta) = \det(P_N\widetilde T^{\,\pm} P_N)\det M^\pm_N,
\end{equation}
where $M^\pm_N$ is the matrix \eqref{eq:def-MN} for some isometry
$V\colon\C^{m^\pm}\to K^\pm$ with $m^\pm:=\dim K^\pm$. An
explicit $V$ adapted to the edge eigenvectors is constructed in the proof of Theorem~\ref{edge-asymp}.

We call $\det(P_N\widetilde T^{\,\pm} P_N)$ the \emph{bulk determinant} and $\det( M^\pm_N)$ the \emph{edge determinant}. If $\delta \geq -\frac 12$ there are no edge eigenvalues, so $m^\pm=0$ and the edge determinant is trivially $1$ and can be ignored. We first determine the asymptotic behavior of the bulk determinant.

\section{The bulk determinant}

We prove the following
\begin{theorem}[Bulk asymptotics]\label{bulk-asymp}
Let $\delta\in\R\setminus(\mathcal P\cup\mathcal C)$, and let
$\widetilde T^{\,\pm}$ be as in \eqref{eq:Ttilde-pm}. Then, as $N\to\infty$,
\begin{equation}\label{eq:bulk-asymptotic}
 \log\det(P_N\widetilde T^{\,\pm}P_N)
 =-\frac{\delta_{\text{reg}}^2\pm\delta_{\text{reg}}}{2}\log N+O_\delta(1).
\end{equation}
\end{theorem}

We start with two auxiliary lemmas.

\begin{lemma}[Trace-class difference]\label{lem:trace-class-shift}
If $\alpha,\beta\in\mathbb R\setminus\{0,-1,-2,\ldots\}$, then
\begin{equation}\label{eq:trace-class-difference}
 \Hh_\alpha-\Hh_\beta\in\cS
\end{equation}
where $\cS$ stands for the trace class.
\end{lemma}

\begin{proof}
Choose $r\in\Nzero$ so large that $\alpha+2r>0$ and $\beta+2r>0$. We
first show that treating $\Hh_{\alpha+2r}-\Hh_{\beta+2r}$ is enough
since the difference has finite rank.
Let $P_r$ be the projection onto the first $r$ coordinates and
$Q_r=\id-P_r$.
Expanding via $\id = P_r + Q_r$ gives
\[
 (\Hh_\alpha-\Hh_\beta) - Q_r(\Hh_\alpha-\Hh_\beta)Q_r=P_r(\Hh_\alpha-\Hh_\beta)+Q_r(\Hh_\alpha-\Hh_\beta)P_r,
\]
which has rank at most $2r$.  Moreover, if
$U_r:\ell^2(\Nzero)\to Q_r\ell^2(\Nzero)$ is given by
$U_re_j=e_{j+r}$, then
\[
 U_r^*Q_r(\Hh_\alpha-\Hh_\beta)Q_rU_r
 =
 \Hh_{\alpha+2r}-\Hh_{\beta+2r}.
\]
The difference between $U_r(\Hh_{\alpha+2r}-\Hh_{\beta+2r})U_r^*$ and
$\Hh_\alpha-\Hh_\beta$ is therefore of finite rank.

Assume, without loss of generality, that
$a:=\alpha+2r<b:=\beta+2r$.  The entries of $\Hh_a-\Hh_b$ satisfy
\[
 \frac1{j+k+a}-\frac1{j+k+b}
 =\int_0^1 t^{j+k}\bigl(t^{a-1}-t^{b-1}\bigr)\,dt.
\]
Hence, for $\psi\in\ell^2(\Nzero)$ with finite support we obtain
\begin{align}
    \<\psi, (\Hh_a - \Hh_b)\psi \>
    &= \int_0^1\bigl(t^{a-1}-t^{b-1}\bigr) \sum_{j,k} \overline{\psi_j} t^j t^k \psi_k\,dt \\
    &= \int_0^1\bigl(t^{a-1}-t^{b-1}\bigr) \abs{\< \psi, t_{\text{vec}} \>}^2\,dt \geq 0,
\end{align}
where $(t_{\text{vec}})_j = t^j$. Since
$t_{\text{vec}}\in\ell^2(\Nzero)$ for $t\in [0,1)$, this extends to
all $\psi\in\ell^2(\Nzero)$ and so the difference is positive. Its
diagonal sum is also finite:
\[
\Tr(\Hh_a - \Hh_b) =
 \sum_{n=0}^\infty
 \left(\frac1{2n+a}-\frac1{2n+b}\right)<\infty.
\]
A positive bounded operator with finite diagonal sum is trace class.
\end{proof}

\begin{lemma}[Determinant bound for trace-class perturbations]\label{lem:stable-traceclass}
Let $R,S$ be bounded self-adjoint operators on a Hilbert space with $S-R\in\cS$. Suppose that
$R\ge c\id$ and $S\ge c\id$ for some $c>0$.  Then, for every $N\in\mathbb N$,
\begin{equation}\label{eq:stable-determinant}
 \abs*{\log\det(P_NSP_N)-\log\det(P_NRP_N)}
 \le c^{-1}\norm{S-R}_1.
\end{equation}
The determinants are taken on the finite-dimensional space $P_N\ell^2(\Nzero)$.
\end{lemma}

\begin{proof}
Set $R_N:=P_NRP_N$ and $K_N:=P_N(S-R)P_N$.  For $0\le t\le1$,
\[
 R_N+tK_N=(1-t)P_NRP_N+tP_NSP_N\ge cP_N.
\]
Therefore
\[
 \frac{d}{dt}\log\det(R_N+tK_N)
 =\Tr\bigl((R_N+tK_N)^{-1}K_N\bigr),
\]
via Jacobi's determinant formula
and the absolute value of the right-hand side is bounded by
$c^{-1}\norm{K_N}_1\le c^{-1}\norm{S-R}_1$.  Integration over $t\in[0,1]$
proves the claim.
\end{proof}

Now, with $s$ as in \eqref{eq:nonthreshold}, we define the auxiliary operator
\begin{equation}\label{eq:Rpm}
 R^\pm:=\id\pm\frac{s}{\pi}\Hh_1.
\end{equation}

Recall from Section~\ref{sec:bulk-edge} that
$\widetilde T^{\,\pm}\ge1-\abs{s}$. The same is true for $R^\pm$, because
$\pi^{-1}\Hh_1$ has spectrum $[0,1]$ \cite{Magnus}.  Moreover,
\begin{equation}\label{eq:regularized-difference}
 \widetilde T^{\,\pm}-R^\pm
 =\pm\frac{s}{\pi}(\Hh_{1+\delta}-\Hh_1)+\Pi^\pm
 \in\cS
\end{equation}
by Lemma~\ref{lem:trace-class-shift} and $\dim K^\pm < \infty$.
Together with the positive lower bound for $\widetilde T^{\,\pm}$ and
$R^\pm$ this supplies the assumptions of
Lemma~\ref{lem:stable-traceclass} which then implies
\begin{equation}\label{eq:regularized-determinant}
  \log\det(P_N\widetilde T^{\,\pm}P_N) = \log \det (P_N R^\pm P_N) + O(1)
\end{equation}
as $N\to\infty$.

The asymptotics of the right-hand side of
\eqref{eq:regularized-determinant} are the case $\alpha=1$, $\beta=\mp
s$ of \cite[Thm.~1.1]{Otte}, whose exponent equals
$\frac12\gamma(\beta)\log N+O(1)$ with
$\gamma(\beta)=-\pi^{-2}\bigl(\arcsin^2\beta+\pi\arcsin\beta\bigr)$
for $\beta\in(-1,1)$. Since $\frac1\pi\arcsin(\mp s)=\pm\delta_{\text{reg}}$
by \eqref{eq:delta-reg-s}, this reads as follows.

\begin{proposition}[{\cite[Thm.~1.1]{Otte}}]\label{prop:otte}
Let $\delta\in\R\setminus(\mathcal P\cup\mathcal C)$ and let $R^\pm$ be as in
\eqref{eq:Rpm}. Then, as $N\to\infty$,
\begin{equation}\label{eq:regularized-asymptotic}
 \log\det(P_NR^\pm P_N)
 =-\frac{\delta_{\text{reg}}^2\pm\delta_{\text{reg}}}{2}\log N+O_\delta(1).
\end{equation}
\end{proposition}

\begin{proof}[Proof of Theorem \ref{bulk-asymp}]
Combining \eqref{eq:regularized-asymptotic} with
\eqref{eq:regularized-determinant} proves \eqref{eq:bulk-asymptotic}.
\end{proof}

\section{The edge determinant}

In this section we determine the asymptotics of the edge determinant.
\begin{theorem}[Edge asymptotics]\label{edge-asymp}
Let $\delta\in\R\setminus(\mathcal P\cup\mathcal C)$, and let $M_N^\pm$ be
the matrix \eqref{eq:def-MN} of Lemma~\ref{lem:factorization} for
$T=T^\pm$. Then, as $N\to\infty$,
\begin{equation}\label{eq:edge-contribution}
 \log\det M_N^\pm
 =
 \sum_{j\in J^\pm_\delta}(2\delta+2j+1)
 \,
 \log N+O(1).
\end{equation}
For $\delta\geq-\frac12$ the sum is empty and $\det M_N^\pm=1$.
\end{theorem}

We first simplify the edge determinant in the following way:

\begin{lemma}[Edge determinant]\label{lem:edge-determinant}
In the setting of Lemma~\ref{lem:factorization}, let
\begin{equation}\label{eq:loewner-constants}
 c=\norm{\widetilde T^{-1/2}}^{-2},
 \qquad
 C=\norm{\widetilde T^{1/2}}^2.
\end{equation}
Then
\begin{equation}\label{eq:loewner-tail}
 c\,V^*Q_NV\le M_N\le C\,V^*Q_NV
\end{equation}
and if $V^*Q_NV$ is positive definite
\begin{equation}\label{eq:edge-logdet}
 \log\det M_N
 =\log\det(V^*Q_NV)+O(1).
\end{equation}
\end{lemma}

\begin{proof}
Let $B:=\widetilde T^{1/2}$.  Since by assumption $\widetilde T > 0$, $B$ is boundedly
invertible.  Moreover, $\widetilde T=\id$ on $K$, so
\begin{equation}\label{eq:B-on-kernel}
 BV=V.
\end{equation}
Let
\[
 C_N:=BP_N: P_N\ell^2(\Nzero)\rightarrow\ell^2(\Nzero).
\]
We write $P_N\ell^2 := P_N\ell^2(\Nzero)$ for brevity.
Then
\[
 C_N^*C_N=P_N\widetilde TP_N,
 \qquad
 C_N^*V=P_NBV=P_NV.
\]
Because $C_N$ is injective and has finite-dimensional range,
\begin{equation}\label{eq:EN-projection}
 E_N:=C_N(C_N^*C_N)^{-1}C_N^*
\end{equation}
is the orthogonal projection onto
$\Ran C_N=BP_N\ell^2$.  Hence, for every $x\in\C^m$,
\begin{align}
 x^*M_Nx
 &=\norm{Vx}^2-
   \bigl\langle C_N^*Vx,(C_N^*C_N)^{-1}C_N^*Vx\bigr\rangle\notag\\
 &=\langle Vx,(\id-E_N)Vx\rangle\notag\\
 &=\norm{(\id-E_N)Vx}^2\notag\\
 &=\dist\bigl(Vx,BP_N\ell^2\bigr)^2.
 \label{eq:MN-distance}
\end{align}

Using \eqref{eq:B-on-kernel},
\begin{equation*}
 \dist\bigl(Vx,BP_N\ell^2\bigr)
 =\inf_{y\in P_N\ell^2}\norm{Vx-By}
 =\inf_{y\in P_N\ell^2}\norm{B(Vx-y)}.
\end{equation*}
For every $z\in\ell^2$,
\[
 \norm{B^{-1}}^{-1}\norm{z}
 \le\norm{Bz}
 \le\norm{B}\norm{z}.
\]
Using $z = Vx - y$, taking the infimum over $y\in P_N\ell^2$, and
using that $P_N$ is an orthogonal projection gives
\begin{equation}\label{eq:distance-comparison}
 \norm{B^{-1}}^{-1}\norm{Q_NVx}
 \le\dist\bigl(Vx,BP_N\ell^2\bigr)
 \le\norm{B}\norm{Q_NVx}.
\end{equation}
Since
\[
 \norm{Q_NVx}^2=x^*V^*Q_NVx,
\]
combining \eqref{eq:MN-distance} and \eqref{eq:distance-comparison} proves
\eqref{eq:loewner-tail} with the constants in
\eqref{eq:loewner-constants}.

Finally, define $G_N:=V^*Q_NV$. Assuming that $G_N>0$, $G_N$ is invertible. Then conjugating
\eqref{eq:loewner-tail} by $G_N^{-1/2}$ gives
\[
 c\id_m\le G_N^{-1/2}M_NG_N^{-1/2}\le C\id_m.
\]
Therefore
\[
 c^m\le\frac{\det M_N}{\det G_N}\le C^m.
\]
Since $m$ is fixed, this is equivalent to \eqref{eq:edge-logdet}.
\end{proof}

Lemma~\ref{lem:edge-determinant} reduces the asymptotic
problem to the tails of the explicit vectors $u_j$ from
Proposition~\ref{prop:spectral-data}, where $u_j(n)$ is the $n$th
Taylor coefficient of $f_j$ in \eqref{eq:fj}.

\begin{lemma}[Tails of the edge eigenvectors]\label{lem:eigenvector-tails}
Assume $\delta\notin\mathbb Z$ and let $j\in\Nzero$ satisfy
$\delta+j<-\frac12$.  Then
\begin{equation}\label{eq:coeff-asymptotic}
 u_j(n)
 =c_j n^{\delta+j}\bigl(1+O(n^{-1})\bigr),
 \qquad
 c_j:=\frac{2^j}{\Gamma(\delta+j+1)}\ne0.
\end{equation}
This implies in particular $u_j \in \ell^2(\Nzero)$.
Moreover, for any $J\subseteq\{j\in\Nzero:\delta+j<-\frac12\}$ we have
\begin{equation}\label{eq:tail-gram-asymptotic}
 \log\det\left(\sum_{n=N}^\infty u_i(n)u_j(n)\right)_{i,j\in J}
 =
 \sum_{j\in J}(2\delta+2j+1) \log N + O(1)
\end{equation}
as $N\to\infty$ with a constant depending on $J$.
\end{lemma}

\begin{proof}
Recall that $u_j(n)$ are the Taylor coefficients of $f_j(z) = (1+z)^j(1-z)^{-(\delta+j+1)}$, see \eqref{eq:fj}. Since
\[
 (1+z)^j=\sum_{\ell=0}^j\binom j\ell z^\ell
\]
and
\[
(1-z)^{-(\delta+j+1)} = \sum_{l=0}^\infty \frac {\Gamma(l+\delta+j+1) }{\Gamma(\delta + j+1)\Gamma(l+1)} z^l,
\]
the Cauchy product formula gives for indices $n\ge j$,
\[
 u_j(n)
 =\sum_{\ell=0}^j\binom j\ell
 \frac{\Gamma(n-\ell+\delta+j+1)}
 {\Gamma(\delta+j+1)\Gamma(n-\ell+1)}.
\]
The gamma-ratio asymptotic expansion
\[
 \frac{\Gamma(n+a)}{\Gamma(n+b)}
 =n^{a-b}\bigl(1+O(n^{-1})\bigr)
\]
as $n\to\infty$
then yields \eqref{eq:coeff-asymptotic}, because
$\sum_{\ell=0}^j\binom j\ell=2^j$.

For $i,j\in J$, the hypotheses imply $2\delta+i+j<-1$.  Therefore, for $i,j\in J$
\begin{equation}\label{eq:tail-pairing}
 \sum_{n=N}^\infty u_i(n)u_j(n)
 =\frac{c_ic_j}{-(2\delta+i+j+1)}
 N^{2\delta+i+j+1}\bigl(1+o(1)\bigr)
\end{equation}
as $N\to\infty$.
Factoring out $N^{\delta+i+1/2}$ from row $i$ and
$N^{\delta+j+1/2}$ from column $j$, we obtain
\begin{equation}
    \det\left(\sum_{n=N}^\infty u_i(n)u_j(n)\right)_{i,j\in J}
    = \det A \times
    N^{\sum_{j\in J}(2\delta+2j+1)}\bigl(1+o(1)\bigr)
\end{equation}
as $N\to\infty$ with
\begin{equation}\label{eq:Cauchy-limit}
 A:= \left(\frac{c_ic_j}{-(2\delta+i+j+1)}\right)_{i,j\in J}.
\end{equation}
The matrix $A$ is positive definite: it is the Gram matrix of the linearly
independent functions $c_jx^{\delta+j}$ in $L^2(1,\infty)$.  Its determinant
is therefore positive, proving \eqref{eq:tail-gram-asymptotic}.
\end{proof}

\begin{proof}[Proof of Theorem \ref{edge-asymp}]
We may assume $J^\pm_\delta\ne\emptyset$, otherwise $M_N^\pm$ is empty and \eqref{eq:edge-contribution} is trivial. Let $U:\C^m\to\Ker T^\pm$, $m:=|J^\pm_\delta|$, send the standard basis to the basis $(u_j)_{j\in J^\pm_\delta}$ of $\Ker T^\pm$ from \eqref{eq:kernel-basis}, let $G:=U^*U=(\langle u_i,u_j\rangle)_{i,j}$ be its positive definite Gram matrix, and set $V:=UG^{-1/2}$. Then $V^*V=\id_m$ and $\Ran V=\Ker T^\pm$, so $V$ is an isometry admissible in \eqref{eq:def-MN}. By the argument in the proof of Lemma \ref{lem:positivity} no nonzero vector of $\Ker T^\pm$ has finite support, so $V^*Q_NV>0$ for all $N$, and Lemma \ref{lem:edge-determinant} gives
\begin{equation}
    \log\det M_N^\pm =\log\det(V^*Q_NV)+O(1).
\end{equation}
Since $V^*Q_NV=G^{-1/2}(U^*Q_NU)G^{-1/2}$ and $(U^*Q_NU)_{i,j}=\langle u_i,Q_Nu_j\rangle=\sum_{n\ge N}u_i(n)u_j(n)$ (the $u_j(n)$ being real),
\begin{equation}
    \log\det(V^*Q_NV)
    =\log\det\Bigl(\sum_{n=N}^\infty u_i(n)u_j(n)\Bigr)_{i,j\in J^\pm_\delta}-\log\det G ,
\end{equation}
with $\log\det G$ a constant. Since $J^\pm_\delta\subseteq\{j\in\Nzero:\delta+j<-\tfrac12\}$, Lemma \ref{lem:eigenvector-tails} yields
\begin{equation}
 \log\det M_N^\pm
 =
 \sum_{j\in J^\pm_\delta}(2\delta+2j+1)
 \,\log N+O(1).
\end{equation}
\end{proof}

\section{Proof of Theorem \ref{thm:main}}

\begin{proof}
Combining Lemma \ref{lem:factorization}, Theorem \ref{bulk-asymp} and Theorem \ref{edge-asymp}
 yields
\begin{equation}\label{eq:master-formula}
 \log D_N^\pm(\delta)
 =\left[
 -\frac{\delta_{\text{reg}}^2\pm \delta_{\text{reg}}}{2}
 +\sum_{j\in J^\pm_\delta}(2\delta+2j+1)
 \right]\log N+O_\delta(1).
\end{equation}
Formula~\eqref{eq:master-formula} is the unsimplified strong asymptotic.
We note that $J^\pm_\delta \neq \emptyset$ only for $\delta< -\frac 12$.

Now, Lemma \ref{lem:unfolding} below concludes the proof.
\end{proof}

\begin{lemma}\label{lem:unfolding}
Let $\delta\in\R\setminus(\mathcal P\cup\mathcal C)$.
Then
\begin{equation}\label{eq:unfolding-identity}
 -\frac{\delta_{\text{reg}}^2\pm \delta_{\text{reg}}}{2}
 +\sum_{j\in J_\delta^\pm}(2\delta + 2j + 1)
 =-\frac{\theta_\delta^2\pm\theta_\delta}{2}.
\end{equation}
\end{lemma}

\begin{proof}
Since for $\delta \geq -\frac 12$ the index set $J^\pm_\delta$ is empty, there is nothing to prove as $\theta_\delta = \delta_{\text{reg}}$.

Let $\delta < -\frac12$, so that $\theta_\delta=\delta$, and define
\begin{equation}\label{eq:m-def}
 m := \abs{J_\delta^+} + \abs{J_\delta^-} = \abs{\{j\in\Nzero:j<-\delta-\tfrac12\}}
\end{equation}
where $\abs{\,\cdot\,}$ denotes the cardinality of a set.
Then $-\delta\in(m-\frac12,m+\frac12)$, and on this interval
\begin{equation}\label{eq:lift}
 \delta_{\text{reg}}=(-1)^m(\delta+m),
 \qquad\text{equivalently}\qquad
 \delta=(-1)^m\delta_{\text{reg}}-m.
\end{equation}
We index the $m$ edge eigenvectors by $j=0,\ldots,m-1$, with
alternating eigenvalues and put
\[
 \rho_j:=(-1)^j(\delta+j) \quad \text{for} \quad 0\le j \le m.
\]
Then $\rho_0=\delta$, $\rho_m=\delta_{\mathrm{reg}}$ by
\eqref{eq:lift}, and
\[
 \rho_{j+1}=(-1)^{j+1}-\rho_j.
\]
By the definition of $m$ and $J_\delta^\pm$ we have
\[
 J_\delta^\pm=\{0\leq j<m:(-1)^j=\mp1\},
\]
and a direct calculation gives
\[
 -\frac{\rho_j^2\pm\rho_j}{2}
 +\frac{\rho_{j+1}^2\pm\rho_{j+1}}{2}
 =
 \begin{cases}
  2\delta+2j+1,&j\in J_\delta^\pm,\\
  0,&j\notin J_\delta^\pm.
 \end{cases}
\]
Summing over $j=0,\ldots,m-1$ and using
$\rho_0=\delta$ and $\rho_m=\delta_{\mathrm{reg}}$ yields
\[
 -\frac{\delta^2\pm\delta}{2}
 +\frac{\delta_{\mathrm{reg}}^2\pm\delta_{\mathrm{reg}}}{2}
 =
 \sum_{j\in J_\delta^\pm}(2\delta+2j+1),
\]
which is \eqref{eq:unfolding-identity}, since $\theta_\delta=\delta$.
\end{proof}

\section*{Acknowledgments and use of AI tools}

The key estimate \eqref{eq:loewner-tail}
was pointed out to us by ChatGPT~5.6~Pro (OpenAI); all proofs were
checked by the authors.


\begin{thebibliography}{99}

\bibitem{AlemanMontesSarafoleanu}
A. Aleman, A. Montes-Rodr\'iguez, and A. Sarafoleanu,
\emph{The eigenfunctions of the Hilbert matrix},
Constr. Approx. \textbf{36} (2012), no.~3, 353--374.
\href{https://doi.org/10.1007/s00365-012-9157-z}{doi:10.1007/s00365-012-9157-z}.

\bibitem{BasorEhrhardt}
E.~L. Basor and T. Ehrhardt,
\emph{Asymptotic formulas for determinants of a sum of finite Toeplitz and Hankel matrices},
Math. Nachr. \textbf{228} (2001), 5--45.
\href{https://doi.org/10.1002/1522-2616(200108)228:1<5::AID-MANA5>3.0.CO;2-E}{doi:10.1002/1522-2616(200108)228:1<5::AID-MANA5>3.0.CO;2-E}.

% \bibitem{BirmanPushnitski}
% M.~Sh. Birman and A.~B. Pushnitski,
% \emph{Spectral shift function, amazing and multifaceted},
% Integral Equations Operator Theory \textbf{30} (1998), no.~2, 191--199.
% \href{https://doi.org/10.1007/BF01238218}{doi:10.1007/BF01238218}.

\bibitem{DeiftItsKrasovsky}
P. Deift, A. Its, and I. Krasovsky,
\emph{Asymptotics of Toeplitz, Hankel, and Toeplitz+Hankel determinants with Fisher-Hartwig singularities},
Ann. of Math. (2) \textbf{174} (2011), no.~2, 1243--1299.
\href{https://doi.org/10.4007/annals.2011.174.2.12}{doi:10.4007/annals.2011.174.2.12}.

\bibitem{FedeleGebert}
E. Fedele and M. Gebert,
\emph{On determinants identity minus Hankel matrix},
Bull. Lond. Math. Soc. \textbf{51} (2019), no.~4, 751--764.
\href{https://doi.org/10.1112/blms.12271}{doi:10.1112/blms.12271}.

\bibitem{FrankPushnitski}
R.~L. Frank and A. Pushnitski,
\emph{The spectral density of a product of spectral projections},
J. Funct. Anal. \textbf{268} (2015), no.~12, 3867--3894.
\href{https://doi.org/10.1016/j.jfa.2015.03.018}{doi:10.1016/j.jfa.2015.03.018}.

\bibitem{Gebert2015}
M. Gebert,
\emph{The asymptotics of an eigenfunction-correlation determinant for Dirac-$\delta$ perturbations},
J. Math. Phys. \textbf{56} (2015), no.~7, article id.\ 072110.
\href{https://doi.org/10.1063/1.4927335}{doi:10.1063/1.4927335}.

\bibitem{GebertKuttlerMuller}
M. Gebert, H. K\"uttler, and P. M\"uller,
\emph{Anderson's Orthogonality Catastrophe},
Comm. Math. Phys. \textbf{329} (2014), no.~3, 979--998.
\href{https://doi.org/10.1007/s00220-014-1914-3}{doi:10.1007/s00220-014-1914-3}.

\bibitem{GebertKuttlerMullerOtte}
M. Gebert, H. K\"uttler, P. M\"uller, and P. Otte,
\emph{The exponent in the orthogonality catastrophe for Fermi gases},
J. Spectr. Theory \textbf{6} (2016), no.~3, 643--683.
\href{https://doi.org/10.4171/JST/135}{doi:10.4171/JST/135}.

\bibitem{GebertPoplavskyi}
M. Gebert and M. Poplavskyi,
\emph{On pure complex spectrum for truncations of random orthogonal matrices and Kac polynomials},
preprint (2019).
\href{https://arxiv.org/abs/1905.03154}{arXiv:1905.03154 [math.PR]}.

\bibitem{GharakhlooIts}
R. Gharakhloo and A. Its,
\emph{A Riemann-Hilbert approach to asymptotic analysis of Toeplitz+Hankel determinants},
SIGMA \textbf{16} (2020), Paper No.~100, 47~pp.
\href{https://doi.org/10.3842/SIGMA.2020.100}{doi:10.3842/SIGMA.2020.100}.

\bibitem{GharakhlooItsII}
R. Gharakhloo and A. Its,
\emph{A Riemann-Hilbert approach to asymptotic analysis of Toeplitz+Hankel determinants II},
Nonlinearity \textbf{39} (2026), no.~8, 085002.
\href{https://doi.org/10.1088/1361-6544/ae8d79}{doi:10.1088/1361-6544/ae8d79}.

\bibitem{Howland}
J.~S. Howland,
\emph{Spectral theory of self-adjoint Hankel matrices},
Michigan Math. J. \textbf{33} (1986), no.~2, 145--153.
\href{https://doi.org/10.1307/mmj/1029003344}{doi:10.1307/mmj/1029003344}.

\bibitem{KnoerrOtteSpitzer}
H.~K. Kn\"orr, P. Otte, and W. Spitzer,
\emph{Anderson's orthogonality catastrophe in one dimension induced by a magnetic field},
J. Phys. A \textbf{48} (2015), no.~32, 325202.
\href{https://doi.org/10.1088/1751-8113/48/32/325202}{doi:10.1088/1751-8113/48/32/325202}.

\bibitem{KostrykinMakarov}
V. Kostrykin and K.~A. Makarov,
\emph{On Krein's example},
Proc. Amer. Math. Soc. \textbf{136} (2008), no.~6, 2067--2071.
\href{https://doi.org/10.1090/S0002-9939-08-09141-7}{doi:10.1090/S0002-9939-08-09141-7}.

\bibitem{Magnus}
W. Magnus,
\emph{On the spectrum of Hilbert's matrix},
Amer. J. Math. \textbf{72} (1950), no.~4, 699--704.
\href{https://doi.org/10.2307/2372284}{doi:10.2307/2372284}.

\bibitem{Otte}
P. Otte,
\emph{A Szeg\H{o} limit theorem related to the Hilbert matrix},
Rocky Mountain J. Math. \textbf{54} (2024), no.~5, 1447--1472.
\href{https://doi.org/10.1216/rmj.2024.54.1447}{doi:10.1216/rmj.2024.54.1447}.

\bibitem{RosenblumII}
M. Rosenblum,
\emph{On the Hilbert matrix. II},
Proc. Amer. Math. Soc. \textbf{9} (1958), 581--585.
\href{https://doi.org/10.1090/S0002-9939-1958-0099599-2}{doi:10.1090/S0002-9939-1958-0099599-2}.

\bibitem{Silbermann}
B. Silbermann,
\emph{On the spectrum of Hilbert matrix operator},
Integral Equations Operator Theory \textbf{93} (2021), Paper No.~21.
\href{https://doi.org/10.1007/s00020-021-02637-5}{doi:10.1007/s00020-021-02637-5}.

% \bibitem{Wilf}
% H. S. Wilf,
% \emph{Finite sections of some classical inequalities},
% Ergebnisse der Mathematik und ihrer Grenzgebiete, Band~52,
% Springer-Verlag, Berlin--New York, 1970.
% \href{https://doi.org/10.1007/978-3-642-86712-5}{doi:10.1007/978-3-642-86712-5}.

\end{thebibliography}
\end{document}